\documentclass[
    letterpaper,
    11pt,
]{article}

\usepackage[colorlinks=true,linkcolor=blue,citecolor = blue]{hyperref}
\usepackage[ruled,lined,algo2e,noend, linesnumbered]{algorithm2e}

\usepackage{macros}
\usepackage{bbm}
\usepackage[numbers]{natbib}
\usepackage[margin=1in]{geometry}
\usepackage{graphicx}
\usepackage{empheq}
\usepackage[T1]{fontenc}
\usepackage{enumitem}
\usepackage{svg}
\usepackage{thm-restate}

\title{Query-Efficient Algorithm to Find all Nash Equilibria in a Two-Player Zero-Sum Matrix Game}
\author{Arnab Maiti \and Ross Boczar \and Kevin Jamieson \and Lillian J. Ratliff
\\\and University of Washington\thanks{
\texttt{\{arnabm2,rjboczar,ratliffl\}@uw.edu}, \texttt{jamieson@cs.washington.edu}}
}
\date{}
\begin{document}

\maketitle
\begin{abstract}
We study the query complexity of finding the set of all Nash equilibria $\calX_\star \times \calY_\star$ in two-player zero-sum matrix games. \citet{fearnley2016finding} showed that for any randomized algorithm, there exists an $n \times n$ input matrix where it needs to query $\Omega(n^2)$ entries in expectation to compute a \emph{single} Nash equilibrium. On the other hand, \citet{bienstock1991note} showed that there is a special class of matrices for which one can query $O(n)$ entries and compute its set of all Nash equilibria. However, these results do not fully characterize the query complexity of finding the set of all Nash equilibria in two-player zero-sum matrix games.

In this work, we characterize the query complexity of finding the set of all Nash equilibria $\calX_\star \times \calY_\star$ in terms of the number of rows $n$ of the input matrix $A \in \mathbb{R}^{n \times n}$, row support size $k_1 := |\bigcup\limits_{x \in \calX_\star} \supp(x)|$, and column support size $k_2 := |\bigcup\limits_{y \in \calY_\star} \supp(y)|$. We design a simple yet non-trivial randomized algorithm that, with probability $1 - \delta$, returns the set of all Nash equilibria $\calX_\star \times \calY_\star$ by querying at most $O(nk^5 \cdot \text{polylog}(n / \delta))$ entries of the input matrix $A \in \mathbb{R}^{n \times n}$, where $k := \max\{k_1, k_2\}$. This upper bound is tight up to a factor of $\text{poly}(k)$, as we show that for any randomized algorithm, there exists an $n \times n$ input matrix with $\min\{k_1, k_2\} = 1$, for which it needs to query $\Omega(nk)$ entries in expectation in order to find the set of all Nash equilibria $\calX_\star \times \calY_\star$.

\end{abstract}
\section{Introduction}
\label{sec:introduction}

Von Neumann's celebrated minimax theorem \citep{v1928theorie} initiated the study of two-player zero-sum matrix games by demonstrating the existence of a Nash equilibrium. Subsequently, it was established that two-player zero-sum matrix games are equivalent to linear programming \citep{dantzig1951proof,adler2013equivalence,brooks2021canonical}. With the development of polynomial-time algorithms for linear programming, it became evident that one can compute a Nash equilibrium for a two-player zero-sum matrix game in polynomial time. This implication has motivated key works at the intersection of game theory and theoretical computer science \citep{chen2009settling,daskalakis2007progress,daskalakis2009note,deligkas2022polynomial,babichenko2020informational}.

One of these works was the development of query-efficient algorithms for finding equilibrium solutions. In the context of two-player zero-sum games, a surprising result was demonstrated by \citet{bienstock1991note}. They presented a deterministic algorithm that finds the strict saddle point (if it exists) by querying $O(n)$ entries of the input matrix $A \in \mathbb{R}^{n\times n}$. A strict saddle point is an entry of a matrix that is uniquely the smallest in its row and uniquely the largest in its column. In game theory, such an entry is referred to as a strict Pure Strategy Nash equilibrium. In contrast, \citet{fearnley2016finding} showed that for any randomized algorithm, there exists an $n \times n$ input matrix $A \in \mathbb{R}^{n\times n}$ from which the algorithm needs to query $\Omega(n^2)$ entries in expectation to compute a single Nash equilibrium. This lower bound arises when the Nash equilibrium is a mixed strategy with all rows and columns in the support of this mixed strategy. The results of \citet{bienstock1991note} and \citet{fearnley2016finding} represent two extremes on the spectrum of solution sizes, highlighting a significant gap in the existing literature. Recently, \citet{auger2014sparse} attempted to address this gap by studying a class of matrices with unique Nash equilibrium and binary entries, characterizing the query complexity in terms of the support size of the Nash equilibrium strategy. However, such a characterization for arbitrary matrices in $\mathbb{R}^{n\times n}$ that may have multiple Nash equilibria remains unaddressed in prior work. This leads us to the following question:

\begin{quote}
\emph{Can we characterize the query complexity of finding the set of all Nash equilibria of a matrix $A\in\mathbb{R}^{n\times n}$ in terms of an appropriate notion of the Nash equilibria solution size?}
\end{quote}



In this work, we address the above question.
To this end, let us briefly review two-player zero-sum games on an input matrix $A\in \bbR^{n\times n}$.
Let $\simplex_n$ denote the $(n-1)$-dimensional probability simplex.
A pair $(\xstar,\ystar)\in \simplex_n\times \simplex_n$ is a Nash equilibrium of an input matrix $A\in \bbR^{n\times n}$ if and only if the following holds for all $(x,y)\in \simplex_n\times \simplex_n$:
\begin{equation}\label{eq:nash_def}
    \langle x,A\ystar\rangle \leq \langle \xstar,A\ystar\rangle \leq \langle \xstar,Ay\rangle.
\end{equation}
Von Neumann's minimax theorem states that $\max_{x \in \simplex_n} \min_{y \in \simplex_n} \langle x, A y \rangle = \min_{y \in \simplex_n} \max_{x \in \simplex_n} \langle x, A y \rangle$ (denoted as the value of the game $V_A^\star$).
 Thus, $(\xstar,\ystar)$ is a Nash equilibrium of the matrix $A$ if and only if \[\xstar\in \calX_\star:=\arg\max_{x\in \simplex_n}\min_{y\in \simplex_n}\langle x,Ay\rangle\quad\text{and}\quad \ystar\in  \calY_\star:=\arg\min_{y\in \simplex_n}\max_{x\in \simplex_n}\langle x,Ay\rangle.\]  An important property of the sets $\calX_\star$ and $\calY_\star$ is that they are convex, closed, and bounded polytopes (see \cite{bohnenblust1948mathematical}). In any zero-sum game, there is either exactly one Nash equilibrium $(\xstar, \ystar)$ or infinitely many. However, the value $\langle \xstar, A\ystar \rangle$ is unique and equal to $V_A^\star$. We refer the reader to \Cref{appendix:nash-property} for additional properties of Nash equilibrium, as well as for details on how to efficiently represent the sets $\calX_\star$ and $\calY_\star$.

Our work diverges from prior work as we focus on finding the set of \emph{all} Nash equilibria $\calX_\star\times\calY_\star$. This allows us to perform a full strategic analysis of a game, whereas finding only a \emph{single} equilibrium provides merely a partial view. For example, finding the set of all Nash equilibria allows us to identify all non-degenerate parts of the game that contain equilibrium solutions. Such a non-degenerate part is typically a square submatrix with a unique Nash equilibrium. Additionally, finding the set of all Nash equilibria helps us determine which rows (resp. columns) are suboptimal, in the sense that these rows (resp. columns) are not best response strategies to some Nash equilibrium strategy of the column (resp. row) player.
 
 In the rest of this section, we outline the appropriate notion of Nash equilibria solution size, detail our contributions and techniques, discuss related works, and introduce important notations.

\subsection{Appropriate notion of the Nash equilibria solution size}
Recall that $V_A^\star := \max_{x \in \simplex_n} \min_{y \in \simplex_n} \langle x, A y \rangle = \min_{y \in \simplex_n} \max_{x \in \simplex_n} \langle x, A y \rangle$, $\calX_\star := \arg\max_{x \in \simplex_n} \min_{y \in \simplex_n} \langle x, A y \rangle$, and $\calY_\star := \arg\min_{y \in \simplex_n} \max_{x \in \simplex_n} \langle x, A y \rangle$. This work aims to characterize the query complexity of finding the set of all Nash equilibria $\calX_\star \times \calY_\star$ in terms of the dimension $n$ of the input matrix $A \in \mathbb{R}^{n \times n}$ and the support sizes $k_1 := |\bigcup\limits_{x \in \calX_\star} \supp(x)|$ and $k_2 := |\bigcup\limits_{y \in \calY_\star} \supp(y)|$. To motivate why these are appropriate notions of Nash equilibria solution size, we first present the following important result by \citet{bohnenblust1950solutions}.
\begin{lemma}[\citet{bohnenblust1950solutions}]\label{technical-multiple-NE}
    Consider an input matrix $A\in\mathbb{R}^{n\times n}$. Let $I_x=\{j\in[n]: V_A^\star=\langle x, A e_j\rangle\}$ and $I_y=\{i\in[n]: V_A^\star=\langle e_i, A y\rangle\}$, where $e_i$ denotes the vector with a $1$ in the $i$-th coordinate and $0$'s elsewhere. For any $v\in\mathbb{R}^n$, let $\supp(v):=\{i\in[n]\mid v_i\neq 0\}$. Then we have the following:
    \begin{equation*}
        \bigcup\limits_{x\in \calX_\star}\supp(x)=\bigcap\limits_{y\in \calY_\star}I_y \quad \text{and} \quad \bigcup\limits_{y\in \calY_\star}\supp(y)=\bigcap\limits_{x\in \calX_\star}I_x
    \end{equation*}
\end{lemma}

The above lemma suggests that rows not in $\bigcup\limits_{x \in \calX_\star} \supp(x)$ are strictly worse than the rows in $\bigcup\limits_{x \in \calX_\star} \supp(x)$. This is because for all $y^* \in \calY^\star$ and $i \in \bigcup\limits_{x \in \calX_\star} \supp(x)$, we have $V_A^\star = \langle e_i, A y^* \rangle$. In contrast, there exists a $\hat{y} \in \calY^\star$ such that for all $i \notin \bigcup\limits_{x \in \calX_\star} \supp(x)$, $V_A^\star > \langle e_i, A \hat{y} \rangle$. Hence, the number of optimal rows $|\bigcup\limits_{x \in \calX_\star} \supp(x)|$ is an appropriate measure of the solution size for the row player. Using similar logic, we can justify that the number of optimal columns $|\bigcup\limits_{y \in \calY_\star} \supp(y)|$ is an appropriate measure of the solution size for the column player.

Moreover, it is easy to observe that changing the value of the entries $(i,j)$ where $i \notin \bigcup\limits_{x \in \calX_\star} \supp(x)$ and $j \notin \bigcup\limits_{y \in \calY_\star} \supp(y)$ does not affect the set of all Nash equilibria $\calX_\star \times \calY_\star$. This implies that only $O(n \cdot \max\{k_1, k_2\})$ elements are important in determining the set of all Nash equilibria $\calX_\star \times \calY_\star$, which further highlights why $k_1$ and $k_2$ are appropriate notions of the solution size.


\subsection{Our contributions and innovative techniques}
Let $k:=\max\{k_1,k_2\}$, where $k_1 := |\bigcup\limits_{x \in \calX_\star} \supp(x)|$ and $k_2 := |\bigcup\limits_{y \in \calY_\star} \supp(y)|$. Our main technical contribution in this paper is the following upper bound result.
\begin{restatable}{theorem}{upper}\label{thm:algo-result}
    There exists a randomized algorithm that, with probability at least $1-\delta$, queries $O(nk^5\cdot \log n\cdot \log(\frac{n}{\delta}))$ entries of the input matrix $A\in\mathbb{R}^{n\times n}$ and returns the set of all Nash equilibria $\calX_\star\times\calY_\star$.
\end{restatable}

We outline the innovative techniques used to achieve the result above. As discussed in the previous section, the results of \citet{bohnenblust1950solutions} suggest a separation between the sub-matrix with row indices $\bigcup\limits_{x \in \calX_\star} \supp(x)$ and column indices $\bigcup\limits_{y \in \calY_\star} \supp(y)$ (which we will refer to as the optimal sub-matrix) and the rest of the matrix. Our goal is to formalize this separation and therefore leverage it to identify the set of all Nash equilibria while minimizing the number of queries required. In the specific case of strict Pure Strategy Nash Equilibrium (strict PSNE), this separation is inherently present, as a strict PSNE is uniquely the least in its row and uniquely the highest in its column. Therefore, one possible approach involves constructing a large matrix with ${n\choose k_1}$ rows and ${n\choose k_2}$ columns, where each entry is bijectively mapped to a $k_1\times k_2$ sub-matrix of $A$. The values of these entries are determined by their corresponding sub-matrices in $A$, and this large matrix possesses a strict PSNE that maps to the optimal sub-matrix of $A$.

Assuming we can construct such a matrix, it remains unclear how to query the entries of this large matrix without significantly increasing the query complexity of our algorithm. The algorithm proposed by \citet{bienstock1991note} is inadequate for this purpose because it queries entries diagonally. If we applied a similar diagonal querying strategy to our large matrix, we would end up querying all entries of the original matrix $A$. Instead, we propose a randomized algorithm for finding strict PSNE that operates differently from the one described by \citet{bienstock1991note}. Specifically, if we have access to query oracles where a single call returns either the entire row or the entire column, our algorithm can identify the strict PSNE by making $O(\log^2(r))$ oracle calls, where $r$ is the total number of rows and columns in the input matrix with a strict PSNE. We can also design a query oracle for the larger matrix that returns a row or column by querying at most $O(nk)$ entries of $A$. Hence, the query complexity of finding the strict PSNE of the large matrix will be nearly $n\cdot \text{poly}(k)$.

One approach to creating such a large matrix is to compute the value of the game on the sub-matrix mapped to each entry and use that as the entry’s value. While this method ensures that the entry in the resulting matrix that maps to the optimal sub-matrix is a Pure Strategy Nash Equilibrium (PSNE)—an entry that is the smallest in its row and the largest in its column—it may not necessarily be a strict PSNE. However, if we are given a set of values $\calV$ that includes the value of the game $V_A^\star$, we can modify the previously constructed matrix to ensure that the PSNE becomes a strict PSNE. Our randomized algorithm for finding the strict PSNE has the additional advantage that, whenever the matrix contains a PSNE, it will, with high probability, compute a set that includes the value of the game $V_A^\star$.

Now, we describe how to make the PSNE that maps to the optimal sub-matrix (which we will refer to as the optimal PSNE) a strict PSNE. Consider an entry that is in the column of the optimal PSNE and has a value equal to $V_A^\star$. If this entry is not the optimal PSNE, using \Cref{technical-multiple-NE}, we can show that the sub-matrix it maps to lacks a Nash equilibrium with all the rows in its support, although it does have a Nash equilibrium with all the columns in its support. Therefore, if we devise a method to decrease the values of all such entries by a very small amount $\varepsilon$, the optimal PSNE will become strictly greater than all the entries in its column.

Although we do not know which column contains the optimal PSNE, we can decrease the values of entries by a very small amount $\varepsilon$ if the entry value is in the set $\calV$, which includes $V_A^\star$, and the sub-matrix it maps to lacks a Nash equilibrium with all the rows in its support, but does have a Nash equilibrium with all the columns in its support. The values of entries corresponding to the optimal sub-matrix remain unchanged, as there exists a Nash equilibrium strategy with all rows and columns in its support. By modifying the entries in this way, we can ensure that the optimal PSNE is strictly larger than all other entries in its column. Similarly, we can make this PSNE strictly smaller than all other entries in its row. Note that while we modify other entries whose value is not equal to $V_A^\star$, we ensure that the relative ordering with respect to $V_A^\star$ remains unchanged. Consequently, we implicitly construct a large matrix with a strict PSNE that maps to the optimal sub-matrix, and we can find this strict PSNE by querying nearly $n \cdot \text{poly}(k)$ entries of $A$.

In summary, our main innovative technique is to reduce the problem of finding the set of all Nash equilibria to finding a strict PSNE in a large matrix, and then to locate this strict PSNE in a query-efficient manner.

We also provide a simple lower bound to demonstrate that $k = \max\{k_1, k_2\}$ accurately captures the query complexity of the problem.
\begin{restatable}{theorem}{lowerr}\label{thm:lower-multiple-ne}
    Let $\calA$ be any randomized algorithm that correctly computes the set of all Nash equilibria $\calX_\star\times \calY_\star$.
Let $\tau(M)$ denote the number of entries queried from the matrix $M$ by $\calA$.
Then, there exists an input matrix $A\in \bbR^{n\times n}$ with $\max\{k_1,k_2\}=k$ and $\min\{k_1,k_2\}=1$ such that $\bbE[\tau(A)]\geq \frac{(n-1)(k+1)}{2}$, where the randomness in $\tau(A)$ is due to $\calA$ only.
\end{restatable}

This result is not surprising, as changing the elements in the rows $\bigcup\limits_{x \in \calX_\star} \supp(x)$ and columns $\bigcup\limits_{y \in \calY_\star} \supp(y)$ can alter the set of all Nash equilibria, making it necessary to observe them.

\subsection{Related work}\label{sec:relatedwork}
Unlike the computational complexity literature, which focuses on minimizing the number of computational steps for algorithms that compute equilibrium while knowing the payoffs, the works on query complexity consider a scenario where an algorithm aims to compute an equilibrium for a game without knowing the payoffs beforehand. The main focus of such works is to minimize the number of payoff queries used. We discuss some of these works here.

\textbf{Query complexity of zero-sum games:} Recently, there has been renewed interest in the topic of query complexity of zero-sum games. Under the assumption that input matrix $A\in\{0,1\}^{n\times n}$ has a unique Nash equilibrium, \citet{auger2014sparse} designed a randomized algorithm for finding the unique Nash equilibrium of $ A $ by querying $ O(n\log(n)\cdot k^{4k}) $ entries of $A$, where $ k = |\supp(x^\star)| $. Recall that for the special case when the input matrix has a strict PSNE, \citet{bienstock1991note} presented a deterministic algorithm that finds the strict PSNE by querying $ O(n) $ entries of the input matrix $ A \in \mathbb{R}^{n \times n} $. Recently, \citet{dallant2024finding} improved this result by designing a deterministic algorithm that finds the strict PSNE in $ O(n\log^*n) $ time by querying $ O(n) $ entries of the input matrix $ A \in \mathbb{R}^{n \times n} $. Parallel to our work, under the same special case, \citet{dallant2024optimal} and \citet{maiti2024near} designed randomized algorithms aimed at minimizing the time complexity and the statistical complexity (when the feedback on the payoff matrix is noisy) respectively.


\textbf{Query complexity of general-sum games:}
The query complexity of finding various equilibrium concepts in games—including general-sum multi-player games—has been extensively studied; see the survey by \citet{babichenko2020informational} for a comprehensive overview. For instance, \citet{babichenko2016query} demonstrated that the query complexity of computing an $\varepsilon$-well-supported Nash equilibrium in an $n$-player binary action game is $2^{\Omega(n)}$. 

\citet{fearnley2013learning} showed that in an $n \times n$ bimatrix game with payoffs in the range [0,1], a $\frac{1}{2}$-approximate Nash equilibrium can be computed using at most $2n - 1$ payoff queries. They complemented this result by showing that there is a constant $\varepsilon < \frac{1}{2}$ such that the query complexity of computing an $\varepsilon$-approximate Nash equilibrium is more than $2n - 1$.

\citet{goldberg2016bounds} proved that for a game with $n$ players and $m$ actions, the query complexity of finding an $\varepsilon$-approximate coarse correlated equilibrium is $O\left(\frac{nm \log m}{\varepsilon^2}\right)$. A randomized algorithm was designed by \citet{fearnley2016finding} that finds a $\left(\frac{3 - \sqrt{5}}{2} + \varepsilon\right)$-approximate Nash equilibrium in $n \times n$ bimatrix games using $O\left(\frac{n \log n}{\varepsilon^2}\right)$ payoff queries. Later, \citet{deligkas2023polynomial} presented an algorithm that finds a $\frac{1}{2} + \varepsilon$-well-supported Nash equilibrium in $n \times n$ general-sum matrix games using $O\left(\frac{n \log n}{\varepsilon^4}\right)$ payoff queries.

\subsection{Notation} 
We refer to the entry in the $i$th row and $j$th column of a matrix $M$ as entry $(i, j)$ of $M$, and we denote its value by $M_{i,j}$.
Let $e_i$ denote the standard basis vector with the $i$th entry equal to one and all other entries equal to zero. We use $\ones$ to represent the all-ones vector.

Let $[n]:=\{1,2,\ldots,n\}$, and let $\binom{[n]}{k_0}$ denote the set of all subsets of $[n]$ of size $k_0$.
For any $v\in\mathbb{R}^n$, let $\supp(v):=\{i\in[n]\mid v_i\neq 0\}$ denote the support of $v$. 

For any matrix $M\in\bbR^{n\times m}$, let $V_M^\star:=\max_{x\in \simplex_{n}}\min_{y\in\simplex_{m}} \langle x,My\rangle$ denote the value of the zero-sum game on matrix $M$. We say an entry  $(\istar,\jstar)$ is a Pure Strategy Nash equilibrium (in short PSNE) of $M$ if $M_{i,\jstar}\leq M_{\istar,\jstar}\leq M_{\istar,j}$ for all $i\neq \istar,j\neq \jstar$. We say an entry $(\istar,\jstar)$ is a strict PSNE of $M$ if $M_{i,\jstar}< M_{\istar,\jstar}< M_{\istar,j}$ for all $i\neq \istar,j\neq \jstar$.

Let \(\mathbb{R}^{\mathcal{I}}\) represent the space of all real-valued vectors indexed by the set of indices \(\mathcal{I}\). The simplex over a set of indices \(\mathcal{I}\), denoted \(\simplex_{\mathcal{I}}\), is defined as the set of vectors \(x \in \mathbb{R}^{\mathcal{I}}\) such that \(\sum_{i \in \mathcal{I}} x_i = 1\) and \(x_i \geq 0\) for all \(i \in \mathcal{I}\). 

For simplicity, we use $\log x$ to denote $\log_2 x$. Also for simplicity, unless otherwise noted, by \textit{number of entries queried}, we mean the number of distinct entries queried.


\paragraph{Outline of the paper:} In \Cref{appendix:nash-property}, we discuss key properties of Nash equilibria. In \Cref{sec:technical}, we present important technical lemmas that will be utilized throughout the paper. \Cref{sec:upper} introduces the randomized algorithm for finding strict PSNE, while \Cref{sec:non-unique} details the algorithm for finding the set of all Nash equilibria and provides a lower bound on query complexity.

\section{Properties of Nash Equilibrium }\label{appendix:nash-property}
 We now state some properties of Nash equilibrium (see \cite{karlin2017game} and \cite{bohnenblust1948mathematical} for more details). Consider $A\in\mathbb{R}^{n\times n}$. 
 The pair $(\xstar,\ystar)$ is a Nash equilibrium of $A$ if and only if the following hold simultaneously:
 \begin{itemize}
    \item The value of the game on $A$ satisfies $\langle e_i, A\ystar\rangle =V^\star_A$ for any $i\in \supp(\xstar)$ and  $\langle \xstar, Ae_j\rangle=V^\star_A$ for any $j\in \supp(\ystar)$.
    \item  The value of the game on $A$ satisfies  $\langle e_i, A\ystar\rangle \leq V^\star_A$ for any $i\notin \supp(\xstar)$ and $\langle \xstar, Ae_j\rangle\geq V^\star_A$ for any $j\notin \supp(\ystar)$.
 \end{itemize}
 It is also easy to observe that $(\xstar,\ystar)$ is a Nash equilibrium of $A$ if and only if $\langle \xstar,A\ystar\rangle=\max_{i\in[n]}\langle e_i,A\ystar\rangle=\min_{j\in[n]}\langle \xstar,Ae_j\rangle$.
 Now, recall that due to Von Neumann's minimax theorem, we have \[V_A^\star:=\max_{x\in \simplex_n}\min_{y\in \simplex_n}\langle x,Ay\rangle=\min_{y\in \simplex_n}\max_{x\in \simplex_n}\langle x,Ay\rangle.\] Observe that this is also equivalent to \[V_A^\star=\max_{x\in \simplex_n}\min_{j\in [n]}\langle x,Ae_j\rangle=\min_{y\in \simplex_n}\max_{i\in[n]}\langle e_i,Ay\rangle.\] Also recall that $(\xstar,\ystar)$ is a Nash equilibrium of the matrix $A$ if and only if \[\xstar\in \calX_\star:= \arg\max_{x\in \simplex_n}\min_{y\in \simplex_n}\langle x,Ay\rangle\quad \text{and} \quad\ystar\in \calY_\star:=\arg\min_{y\in \simplex_n}\max_{x\in \simplex_n}\langle x,Ay\rangle.\] 

This implies that if there are strategies $\xstar,\ystar\in\simplex_n$ such that $\min_{y\in \simplex_n}\langle \xstar,Ay\rangle=V_A^\star$ and $\max_{x\in \simplex_n}\langle x,A\ystar\rangle=V_A^\star$, then $\xstar\in\calX_\star$ and $\ystar\in\calY_\star$. Also, convex combination of multiple Nash equilibria is a Nash equilibrium as $\calX_\star,\calY_\star$ are compact convex sets.
 

The sets $\calX_\star$ and $\calY_\star$ can be efficiently represented as the convex hulls of their extreme points, which are finite and non-zero in number. Let $\calE(\calX)$ denote the set of extreme points of a set $\calX$. For any submatrix $\tilde{A}$ of $A$, define $\tilde{x}$ and $\tilde{y}$ as the vectors formed by removing the components of $x$ and $y$ that correspond to the rows and columns missing from $\tilde{A}$. The following lemma describes an important property of the extreme points of $\calX_\star$ and $\calY_\star$.
\begin{lemma}[\citet{bohnenblust1948mathematical}]\label{basic-solution}
   A Nash equilibrium $(x_\star,y_\star) $ of $A$ is in the set $ \calE(\calX_\star)\times \calE(\calY_\star)$ if and only if there is a square sub-matrix $\tilde{A}$ of $A$ such that 
    \begin{equation*}
        \tilde{x}^\top_\star=\frac{\mathbf{1}^\top adj(\tilde{A})}{\mathbf{1}^\top adj(\tilde{A})\mathbf{1}},\; \tilde{y}_\star=\frac{adj(\tilde{A})\mathbf{1}}{\mathbf{1}^\top adj(\tilde{A})\mathbf{1}},\; V_A^\star=\frac{det(\tilde{A})}{\mathbf{1}^\top adj(\tilde{A})\mathbf{1}}
    \end{equation*}
\end{lemma}

It is easy to observe that we can compute all the extreme points of $\calX_\star$ and $\calY_\star$ by examining all the square submatrices $\tilde{A}$ of the input matrix $A$ and applying the above lemma. Let $I_x = \{j \in [n] : V_A^\star = \langle x, A e_j \rangle\}$ and $I_y = \{i \in [n] : V_A^\star = \langle e_i, A y \rangle\}$. Since the sets $\calX_\star$ and $\calY_\star$ are the convex hulls of $\calE(\calX_\star)$ and $\calE(\calY_\star)$, we have $\bigcup\limits_{x_\star \in \calX_\star} \supp(\xstar) = \bigcup\limits_{x_\star \in \calE(\calX_\star)} \supp(\xstar)$, $\bigcup\limits_{y_\star \in \calY_\star} \supp(\ystar) = \bigcup\limits_{y_\star \in \calE(\calY_\star)} \supp(\ystar)$, $\bigcap\limits_{x_\star \in \calX_\star} I_x = \bigcap\limits_{x_\star \in \calE(\calX_\star)} I_x$, and $\bigcap\limits_{y_\star \in \calY_\star} I_y = \bigcap\limits_{y_\star \in \calE(\calY_\star)} I_y$. This fact allows us to verify conditions involving $I_x$, $I_y$, and the union of supports, which is critical to our algorithm in later sections.
\section{Technical Lemmas for Nash Equilibrium}\label{sec:technical}

Recall that for the input matrix $A \in \mathbb{R}^{n \times n}$, $\calX_\star := \arg\max_{x \in \simplex_n} \min_{y \in \simplex_n} \langle x, Ay \rangle$ and $\calY_\star := \arg\min_{y \in \simplex_n} \max_{x \in \simplex_n} \langle x, Ay \rangle$. Also recall that $\simplex_{\calI}$ represents the simplex over a set of indices $\calI$. For any vector $v$ and the set of indices $\calS = \{i_1, \ldots, i_{\ell}\}$, let $v_\calS$ denote the vector $(v_{i_1}, \ldots, v_{i_{\ell}})$. We now state an important lemma that provides the necessary and sufficient condition for a subset of rows and columns to be optimal.
\begin{lemma}\label{technical:correctness}
    Let $\calI$ and $\calJ$ be two subsets of $[n]$. Let $M_1$ be the submatrix of $A$ with row indices $\calI$ and column indices $[n]$, and $M_2$ be the submatrix of $A$ with row indices $[n]$ and column indices $\calJ$. Let $\calX_{M_1} := \arg\max_{x \in \simplex_{\calI}} \min_{y \in \simplex_{n}} \langle x, M_1 y \rangle$ and $\calY_{M_2} := \arg\min_{y \in \simplex_{\calJ}} \max_{x \in \simplex_{n}} \langle x, M_2 y \rangle$. Let $\calX_A := \{x \in \simplex_n : x_\calI \in \calX_{M_1} \}$ and $\calY_A := \{y \in \simplex_n : y_\calJ \in \calY_{M_2} \}$. The following holds:
\begin{enumerate}
    \item If $\calI=\bigcup\limits_{x\in\calX_\star}\supp(x)$ and $\calJ=\bigcup\limits_{y\in\calY_\star}\supp(y)$, then the following holds:
    \begin{itemize}
        \item $\calX_A=\calX_\star$ and $\calY_A=\calY_\star$
        \item $V_{M_1}^\star=V_A^\star$ and $\bigcap\limits_{x\in\calX_A}I_x=\calJ$ where $I_x=\{j\in[n]: V_A^\star=\langle x, A e_j\rangle\}$. 
        \item $V_{M_2}^\star=V_A^\star$ and $\bigcap\limits_{y\in\calY_A}I_y=\calI$ where $I_y=\{i\in[n]: V_A^\star=\langle e_i, A y\rangle\}$.
    \end{itemize}
    \item Let $I_x'=\{j\in[n]: V_{M_1}^\star=\langle x, A e_j\rangle\}$ and $I_y'=\{i\in[n]: V_{M_2}^\star=\langle e_i, A y\rangle\}$. If $V_{M_1}^\star=V_{M_2}^\star$, $\bigcap\limits_{x\in\calX_A}I_x'=\calJ$ and $\bigcap\limits_{y\in\calY_A}I_y'=\calI$ then  $\calX_A=\calX_\star$ and $\calY_A=\calY_\star$.
\end{enumerate}
\end{lemma}
\begin{proof}
    Let us assume that $\calI=\bigcup\limits_{x\in\calX_\star}\supp(x)$ and $\calJ=\bigcup\limits_{y\in\calY_\star}\supp(y)$. First consider $\hat x\in \calX_A$. Observe that $\supp(\hat x)\subseteq\calI$. Now observe that $V_{M_1}^\star=\min_{y\in\simplex_n}\langle \hat x,Ay\rangle\leq \max_{x\in\simplex_n}\min_{y\in\simplex_n}\langle  x,Ay\rangle=V_A^\star$. Next consider $x^\star\in\calX_\star$. Observe that $\supp(x^\star)\subseteq \calI$. Now observe that $V_A^\star=\min_{y\in\simplex_n}\langle x^\star_\calI,M_1y\rangle\leq \max_{x\in\simplex_{\calI}}\min_{y\in\simplex_n}\langle  x,M_1y\rangle=V_{M_1}^\star$. Hence, we have $V_A^\star=V_{M_1}^\star$ and therefore $\hat x\in\calX_\star$ and $x^\star\in\calX_A$ implying that $\calX_A=\calX_\star$. Analogously, one can show that $V_{M_2}^\star=V_A^\star$ and $\calY_A=\calY_\star$. Due to \Cref{technical-multiple-NE}, we get that $\bigcap\limits_{x\in\calX_A}I_x=\bigcap\limits_{x\in\calX_\star}I_x=\bigcup\limits_{y\in\calY_\star}\supp(y)=\calJ$ and $\bigcap\limits_{y\in\calY_A}I_y=\bigcap\limits_{y\in\calY_\star}I_y=\bigcup\limits_{x\in\calX_\star}\supp(x)=\calI$.

    Next let us assume that $V_{M_1}^\star=V_{M_2}^\star$, $\bigcap\limits_{x\in\calX_A}I_x'=\calJ$ and $\bigcap\limits_{y\in\calY_A}I_y'=\calI$. Consider $\hat x\in \calX_A$ and $\hat y\in\calY_A$. Observe that $\supp(\hat x)\subseteq\calI$ and $\supp(\hat y)\subseteq\calJ$. Now observe that $\langle \hat x,A\hat y\rangle\geq \min_{y\in\simplex_n}\langle \hat x,Ay\rangle=V_{M_1}^\star$. Next observe that $\langle \hat x,A\hat y\rangle\leq \max_{x\in\simplex_n}\langle x,A\hat y\rangle=V_{M_2}^\star$. Since $V_{M_1}^\star=V_{M_2}^\star$, we have $\langle \hat x,A\hat y\rangle=V_{M_1}^\star=V_{M_2}^\star$. This implies that $\langle \hat x,A\hat y\rangle=V_{M_2}^\star= \max_{x\in \simplex_n}\langle x,A\hat y\rangle$ and $\langle \hat x,A\hat y\rangle=V_{M_1}^\star= \min_{y\in \simplex_n}\langle \hat x,A y\rangle$. Hence $(\hat x,\hat y)$ is a Nash equilibrium of $A$ and $V_A^\star=\langle \hat x,A\hat y\rangle=V_{M_1}^\star=V_{M_2}^\star$. This implies that $\calX_A\subseteq \calX_\star$ and $\calY_A\subseteq \calY_\star$. Consider $x^\star\in \calX_\star$. Due to \Cref{technical-multiple-NE}, we have $\supp(x^\star)\subseteq\bigcap\limits_{y\in\calY_\star}I_y\subseteq\bigcap\limits_{y\in\calY_A}I_y'=\calI$. Now we have $\min_{y\in\simplex_n}\langle x^\star_{\calI},M_1y\rangle=\min_{y\in\simplex_n}\langle x^\star,Ay\rangle=V_A^\star=V_{M_1}^\star$. Hence, $x^\star_{\calI}\in \calX_{M_1}$ which implies $\calX_\star\subseteq \calX_A$. Analogously, one can show that $\calY_{\star}\subseteq \calY_A$. Hence, we have $\calX_A=\calX_\star$ and $\calY_A=\calY_\star$.
\end{proof}

The following lemma states that a matrix can have at most one strict PSNE.
\begin{lemma}\label{Unique-PSNE}
Any matrix $M\in\bbR^{n\times m}$ can have at most one strict PSNE.
\end{lemma}
\begin{proof}
    Consider a matrix $M$ with a strict PSNE at $(i_\star,j_\star)$. Since $(i_\star,j_\star)$ is a strict PSNE, we have $M_{i,j_\star} < M_{i_\star,j}$ for all $(i,j)\neq (\istar,\jstar)$. Suppose there is another entry $(i',j') \neq (i_\star,j_\star)$ that is also a strict PSNE. Since $(i',j')$ is a PSNE, we would have $M_{i_\star,j'} \leq M_{i',j'} \leq M_{i',j_\star}$. However, this contradicts the inequality $M_{i',j_\star} < M_{i_\star,j'}$. Therefore, the matrix $M$ cannot have more than one strict PSNE.
    
\end{proof}

\section{An important subroutine for PSNE}\label{sec:upper}

A key part of our algorithm for finding the set of all Nash equilibria is a query-efficient subroutine for a matrix $M\in\mathbb{R}^{n\times m}$ that has a PSNE and potentially a strict PSNE. Moreover, we assume that we have access to the matrix $M$ through a query oracle $\calQ$ which operates as follows:
\begin{itemize}
    \item A single call to query oracle $\calQ$ with input $(\text{row},i)$ returns the entire $i$-th row of $M$.
    \item A single call to query oracle $\calQ$ with input $(\text{column},j)$ returns the entire $j$-th column of $M$.
\end{itemize}

\noindent
In this section, we work with a matrix $M\in\mathbb{R}^{n\times m}$ and we aim to achieve the following goals by making $O(\log^2(n+m))$ calls to the query oracle $\calQ$:
\begin{itemize}
    \item If the matrix $M$ has a PSNE, we construct a set of distinct real numbers that contains $V_M^\star$ with high probability.
    \item If the matrix $M$ has a strict PSNE, we find the strict PSNE with high probability.
\end{itemize}

\subsection{A Randomized Algorithm to Find the Strict 
PSNE}\label{sec:random:psne}


In this section, we assume that the input matrix $M\in\bbR^{n\times m}$ has a strict PSNE $(\istar,\jstar)$. By definition, we have that $M_{i,\jstar}<M_{\istar,\jstar}<M_{\istar,j}$ for all $i\in [n]\setminus \{\istar\}$ and $j\in[m]\setminus \{\jstar\}$.
We use this crucial fact to design a randomized algorithm, \psnealg/  (\Cref{alg:psne}), for finding the strict PSNE.

We first describe our procedure, \psnealg/, intuitively. The procedure progresses through a logarithmic number of stages. In each stage, it aims to eliminate half of the sub-optimal rows and columns. At any given stage where the goal is to eliminate half of the sub-optimal rows, we begin by sampling a logarithmic number of row indices and querying all entries in each sampled row. Next, we examine the highest values among the queried entries in each column and select the column $\widehat{j}$ with the lowest of these values. We then sample all entries in column $\widehat{j}$ and remove rows where the corresponding entries in the column $\widehat{j}$ are below the median value of that column. An analogous process exists to eliminate half of the sub-optimal columns. This process of sampling, querying, and reducing the search space continues until only one row and one column remain, which together form the strict PSNE. For a formal description of the procedure \psnealg/, please refer to \Cref{alg:psne}.

\begin{algorithm2e}[!ht]
\caption{\psnealg/($M,\delta$)}
\label{alg:psne}
\DontPrintSemicolon
$n\gets$ number of rows of $M$, $m\gets$ number of columns of $M$,\:\: $\calX_1 \gets [n], \:\calY_1 \gets [m], \:\calV_1\gets\emptyset$\;
\For{$t=1,2,\ldots$}{
    \If{$|\calX_t|=1$ and $|\calY_t|=1$}{
    $\calV\gets \calV_{t}\cup\{M_{\ihat,\jhat}\}$, where $\calX_t=\{\;\ihat\;\}$ and $\calY_t=\{\;\jhat\;\}$\;
    \KwRet $(\;\ihat,\:\jhat\;)$ and $\calV$
    }
    \If{$|\calX_t|=1$ }{
    Query all the elements of row $\ihat$, where $\calX_t=\{\;\ihat\;\}$\;
    Choose the argument $\jhat\in \argmin_{j\in \calY_t}M_{\ihat,j}$ and set $\calV\gets \calV_t\cup \{M_{\ihat,\jhat}\}$ \;
    \KwRet $(\;\ihat,\:\jhat\;)$ and $\calV$
    }
    \If{$|\calY_t|=1$}{
    Query all the elements of column $\jhat$, where $\calY_t=\{\;\jhat\;\}$\;
    Choose the argument $\ihat\in \argmax_{i\in \calX_t}M_{i,\jhat}$ and set $\calV\gets \calV_t\cup \{M_{\ihat,\jhat}\}$ \;
    \KwRet $(\;\ihat,\:\jhat\;)$ and $\calV$
    }
    $n_{t}\gets |\calX_t|$, $m_t\gets |\calY_t|$ and $\ell\gets \left\lceil 2\log\left(\frac{2(n+m)^2}{\delta}\right)\right\rceil$\;
    $\Xbar_t \gets\{x_1,x_2,\ldots, x_{\ell}\}$; $x_i$'s are drawn independently and uniformly at random from $\calX_t$ \label{recurse:line2}\;
    $\Ybar_t \gets \{y_1,y_2,\ldots, y_{\ell}\}$; $y_j$'s are drawn independently and uniformly at random from $\calY_t$ \label{recurse:line3}\;
    Query all the entries of $M$ in $\{(i,j):i\in\calX_t,j\in \Ybar_t\}\cup\{(i,j):i\in \Xbar_t,j\in \calY_t\}$\label{query-st:1}\;
    
    Choose the arguments $\ihat\in \argmax_{i\in\calX_t}\min_{j\in \Ybar_t}M_{i,j}$ and $\jhat\in \argmin_{j\in\calY_t}\max_{i\in \Xbar_t}M_{i,j}$\;
    $\calV_{t+1}\gets \calV_t\cup\{\max_{i\in\calX_t}\min_{j\in \Ybar_t}M_{i,j},\min_{j\in\calY_t}\max_{i\in \Xbar_t}M_{i,j}\}$\;
    Query all the entries of $M$ in $\{({i,\jhat}): i\in \calX_t\}\cup\{({\ihat,\:j}):j\in\calY_t\}$\label{query-st:2}\;
    
    $\calX_{t+1}\gets \{i_1,i_2,\ldots, i_{\lfloor n_t/2\rfloor}\} \subseteq \calX_t=\{i_1,\ldots,i_{n_t}\}$,   where $M_{i_1,\jhat}\geq M_{i_2,\jhat}\geq \ldots \geq M_{i_{n_t},\jhat}$\;
    $\calY_{t+1}\gets \{j_1,j_2,\ldots, j_{\lfloor m_t/2\rfloor}\}\subseteq \calY_t=\{j_1,\ldots,j_{m_t}\}$, where $M_{\ihat,j_1}\leq M_{\ihat,j_2}\leq \ldots \leq M_{\ihat,j_{m_t}}$\;
}
\end{algorithm2e}


We now begin the analysis of \Cref{alg:psne} with the following proposition.

\begin{proposition}\label{lem:probability-bound}
Fix $c,k\in \bbN$ with $k\geq 2$.
Consider a set $S=\{a_1,a_2,\ldots,a_{k}\}$, and let $S_1=\{a_1,a_2,\ldots, a_{\lfloor k/2\rfloor}\}$.
Let $\ell=\lceil 2\log (\frac{c}{\delta})\rceil$, and let $X=\{x_1,x_2,\ldots,x_\ell\}$ where $x_i$'s are drawn independently and uniformly at random from $S$.
Then, $X\cap S_1 \neq \emptyset$ with probability at least $1-\frac{\delta}{c}$.
\end{proposition}
\begin{proof} As the events $\{x_i \notin S_1\}$ are independent,
\begin{align*}
\bbP(X\cap S_1 = \emptyset) = \prod_{i=1}^\ell \bbP(x_i \notin S_1) = \Par*{\tfrac{k-\lfloor k/2 \rfloor}{k}}^\ell \leq \Par*{\tfrac{2}{3}}^{2\log (\frac{c}{\delta})}<\Par*{\tfrac{1}{2}}^{\log (\frac{c}{\delta})}=\tfrac{\delta}{c}.
\end{align*}
\end{proof}

We now establish the query complexity of \psnealg/ in the following lemma.
\begin{lemma}\label{lem:psne-query}
\psnealg/$(A,\delta)$ can be executed by calling the query oracle $\calQ$ at most $O(\log(n+m)\cdot\log(\frac{n+m}{\delta}))$ times.
\end{lemma}
\begin{proof}
Observe that the algorithm terminates at a fixed iteration $t_\star$ which is at most $\min\{\log(n),\log(m)\}+1$. In each iteration $t<t_\star$, line \ref{query-st:1} requires at most $2\ell$ calls to the query oracle $\calQ$ and line \ref{query-st:2} requires two calls to the query oracle $\calQ$. Final iteration $t_\star$ requires at most one call to the query oracle $\calQ$. Hence, \psnealg/$(A,\delta)$ can be executed by calling the query oracle $\calQ$ at most $O(\log(n+m)\cdot\log(\frac{n+m}{\delta}))$ times.
\end{proof}

Finally, we establish the correctness of \psnealg/ in the following lemma.
\begin{lemma}\label{lem:psne-correct}
If $M$ has a strict PSNE, then \psnealg/$(M,\delta)$ returns the strict PSNE with probability at least $1-\delta$.
\end{lemma}

\begin{proof}
Let $(\istar,\jstar)$ be the strict PSNE of $M$. Consider an iteration $t$ such that $\min\{|\calX_t|,|\calY_t|\}\geq 2$.
Let us assume that $\istar\in\calX_t$ and $\jstar\in\calY_t$. 
Under this assumption, we will show that $\istar\in\calX_{t+1}$ and $\jstar\in\calY_{t+1}$ with probability at least $1-\frac{\delta}{n+m}$. 
    
For each $j\in \calY_t$, let us relabel the indices in $\calX_t$ as $\{i_1,i_2,\ldots,i_{n_t}\}$ such that $M_{i_1,j}\geq M_{i_2,j}\geq\ldots\geq M_{i_{n_t},j}$ and define $\calR_{t,j}:=\{i_1,i_2,\ldots,i_{\lfloor n_t/2\rfloor}\}$. 
Recall that $\Xbar_t$ is a random subset of $\calX_t$. By \Cref{lem:probability-bound} we have $\mathbb{P}(\Xbar_{t}\cap \calR_{t,j}\neq \emptyset)\geq 1-\frac{\delta}{2(n+m)^2}$ for a fixed $j$.
By the union bound, we then have that with probability at least $1-\frac{\delta}{2(n+m)}$, the event $\calE_t:=\{\Xbar_t\cap \calR_{t,j}\neq \emptyset \:\forall \:j\in\calY_t\}$ holds. Note that if $\calE_t$ holds, for each $j\in\calY_t$ we have $\argmax_{i\in\Xbar_t}M_{i,j}\cap \calR_{t,j}\neq \emptyset$.

Assume $\calE_t$ holds, and let $v^\star_y=\min_{j\in\calY_t}\max_{i\in \Xbar_t}M_{i,j}$. Recall that $\jhat\in \arg\min_{j\in\calY_t}\max_{i\in \Xbar_t}M_{i,j} $. If $\jhat=\jstar$, then $\istar\in  \calR_{t,\jhat}=\calX_{t+1}$, as $M_{\istar,\jstar}$ is the unique highest element of column $\jstar$. Next, let us consider the case where $\jhat\neq \jstar$.
As $\calE_t$ holds, $\argmax_{i\in \Xbar_t}M_{i,\jhat}\cap \calR_{t,\jhat}\neq \emptyset$.
We also see that 
$v^\star_y = \min_{j\in\calY_t}\max_{i\in \Xbar_t}M_{i,j} \leq \max_{i\in \Xbar_t}M_{i,\jstar}\leq M_{\istar,\jstar}$. Recall that for all $j\in\calY_t\setminus\{\jstar\}$, we have $M_{\istar,j}> M_{\istar,\jstar}$.
Now for all $i\in \calX_t\setminus \calR_{t,\jhat}$, we have $M_{i,\jhat}\leq v_y^\star\leq  M_{\istar,\jstar}<M_{\istar,\jhat}$. Hence, $\istar \in \calR_{t,\jhat}=\calX_{t+1}$.

By an identical argument, we have $\jstar\in \calY_{t+1}$ with probability at least $1-\frac{\delta}{2(n+m)}$.
Hence, by the union bound, with probability at least $1-\frac{\delta}{n+m}$ we have $\istar\in\calX_{t+1}$ and $\jstar\in\calY_{t+1}$.
Observe that the algorithm terminates at a fixed iteration $t_\star$ which is at most $\min\{\log(n),\log(m)\}+1$.
Now observe that \psnealg/$(M,\delta)$ returns $(\istar,\jstar)$ if $\istar\in\calX_{\tstar}$ and $\jstar\in\calY_{\tstar}$ as by definition the entry $(\istar,\jstar)$ is the unique least element in the row $\istar$ and unique highest element in the column $\jstar$.
Let $p_t=\bbP(\istar\in \calX_{t},\jstar\in\calY_t|\istar\in \calX_{t-1},\jstar\in\calY_{t-1})$.
Now we have the following due to the chain rule and Bernoulli's inequality,
 \begin{equation*}
     \bbP(\istar\in\calX_{\tstar},\jstar\in\calY_{\tstar})=\prod_{t=2}^{t_\star}p_t\geq \left(1-\frac{\delta}{n+m}\right)^{t_\star}\geq 1-\delta.
 \end{equation*}
\end{proof}


\subsection{Important property of \psnealg/ for Non-Strict PSNE}
Let us assume we are given an input matrix $M$ with a PSNE $(\istar, \jstar)$. By definition, we have $M_{i, \jstar} \leq M_{\istar, \jstar} \leq M_{\istar, j}$ for all $i\in [n] \setminus \{\istar\}$ and $j \in [m] \setminus \{\jstar\}$. In the following lemma, we state an important property of \psnealg/, applicable even when the PSNE is not strict.

\begin{lemma}\label{lem:non-unique-psne}
    Consider an input matrix $M$ which has a PSNE $(\istar,\jstar)$. With probability at least $1-\delta$, $\psnealg/(M,\delta)$ returns a set $\calV$ of distinct real numbers such that $V_M^\star\in\calV$.
\end{lemma}
\begin{proof}

    Let $(\istar,\jstar)$ be a PSNE of $M$. Observe that the algorithm terminates at a fixed iteration $t_\star$ which is at most $\min\{\log(n),\log(m)\}+1$. Consider an iteration $t<t_\star$.
Let us assume that $\istar\in\calX_t$ and $\jstar\in\calY_t$. 
Under this assumption, we will show that with probability at least $1-\frac{\delta}{n+m}$ either $\istar\in\calX_{t+1}$ and $\jstar\in\calY_{t+1}$ or $V_M^\star\in\calV_{t+1}$. 
    
For each $j\in \calY_t$, let us relabel the indices in $\calX_t$ as $\{i_1,i_2,\ldots,i_{n_t}\}$ such that $M_{i_1,j}\geq M_{i_2,j}\geq\ldots\geq M_{i_{n_t},j}$ and define $\calR_{t,j}:=\{i_1,i_2,\ldots,i_{\lfloor n_t/2\rfloor}\}$. 
Recall that $\Xbar_t$ is a random subset of $\calX_t$. By \Cref{lem:probability-bound} we have $\mathbb{P}(\Xbar_{t}\cap \calR_{t,j}\neq \emptyset)\geq 1-\frac{\delta}{2(n+m)^2}$ for a fixed $j$.
By the union bound, we then have that with probability at least $1-\frac{\delta}{2(n+m)}$, the event $\calE_t:=\{\Xbar_t\cap \calR_{t,j}\neq \emptyset \:\forall \:j\in\calY_t\}$ holds. Note that if $\calE_t$ holds, for each $j\in\calY_t$ we have $\argmax_{i\in\Xbar_t}M_{i,j}\cap \calR_{t,j}\neq \emptyset$.

Assume $\calE_t$ holds, and let $v^\star_y=\min_{j\in\calY_t}\max_{i\in \Xbar_t}M_{i,j}$. Recall that $\jhat\in \arg\min_{j\in\calY_t}\max_{i\in \Xbar_t}M_{i,j} $. As $\calE_t$ holds, $\argmax_{i\in \Xbar_t}M_{i,\jhat}\cap \calR_{t,\jhat}\neq \emptyset$.
We also see that 
$v^\star_y = \min_{j\in\calY_t}\max_{i\in \Xbar_t}M_{i,j} \leq \max_{i\in \Xbar_t}M_{i,\jstar}\leq M_{\istar,\jstar}=V_M^\star$. If $v^\star_y=V_M^\star$, then $V_M^\star\in\calV_{t+1}$ as $\min_{j\in\calY_t}\max_{i\in \Xbar_t}M_{i,j}$ is added to $\calV_t$ to create $\calV_{t+1}$. Recall that for all $j\in\calY_t$, we have $M_{\istar,j}\geq M_{\istar,\jstar}$.
Now if $v^\star_y<V_M^\star$, then for all $i\in \calX_t\setminus \calR_{t,\jhat}$, we have $M_{i,\jhat}\leq v_y^\star<  M_{\istar,\jstar}\leq M_{\istar,\jhat}$. Hence, in this case we have $\istar \in \calR_{t,\jhat}=\calX_{t+1}$.

By an identical argument, with probability at least $1-\frac{\delta}{2(n+m)}$, either $\jstar\in \calY_{t+1}$ or $V_M^\star\in\calV_{t+1}$.
Hence, by the union bound, with probability at least $1-\frac{\delta}{n+m}$, either $\istar\in\calX_{t+1}$ and $\jstar\in\calY_{t+1}$ or $V_M^\star\in\calV_{t+1}$. For the last iteration $t_\star$, 
if $\istar\in\calX_{t_\star}$ and $\jstar\in\calY_{t_\star}$, then $V_M^\star\in\calV$ as the entry $(\istar,\jstar)$ is the highest element of column $\jstar$ and the least element of row $\istar$. Let $p_t=\bbP(\istar\in \calX_{t}\text{ and }\jstar\in\calY_t \text{ or } V_M^\star\in\calV_{t}|\istar\in \calX_{t-1}\text{ and }\jstar\in\calY_{t-1} \text{ or } V_M^\star\in\calV_{t-1})$.
Now we have the following due to the chain rule and Bernoulli's inequality,
 \begin{equation*}
     \bbP(V_M^\star\in \calV)=\prod_{t=2}^{t_\star}p_t\geq \left(1-\frac{\delta}{n+m}\right)^{t_\star}\geq 1-\delta.
 \end{equation*}
\end{proof}

\section{Results for finding the set of all Nash equilibria}\label{sec:non-unique}

Consider an input matrix $A\in\mathbb{R}^{n\times n}$. Recall that $\calX_\star:=\arg\max_{x\in\simplex_n}\min_{y\in\simplex_n}\langle x,Ay\rangle$, $\calY_\star:=\arg\min_{y\in\simplex_n}\max_{x\in\simplex_n}\langle x,Ay\rangle$, $k_1:=|\bigcup\limits_{x\in \calX_\star}\supp(x)|$, $k_2:=|\bigcup\limits_{y\in \calY_\star}\supp(y)|$, and $k:=\max\{k_1,k_2\}$. First, in \Cref{sec:random-algo}, we design an algorithm that requires nearly $n \cdot \text{poly}(k)$ queries to identify $\calX_\star$ and $\calY_\star$ with high probability. We complement this result by showing in \Cref{sec:lower-bound} that $\Omega(nk)$ queries are required in expectation to identify $\calX_\star$ and $\calY_\star$ even when $\min\{k_1, k_2\} = 1$.


\subsection{Randomized Algorithm }\label{sec:random-algo}


We begin with restating our main result.
\upper*

Here is the outline of our randomized algorithm. We begin by designing a randomized procedure called \textsc{SetV} in \Cref{sec:setv}, which takes the values $k_1$ and $k_2$ as input and finds a set $\calV$ of value of the game candidates that includes $V_A^\star$. Next, we design a randomized procedure called $\textsc{Nash}$ in \Cref{sec:k1k2-known} that also takes the values $k_1$ and $k_2$ as input, calls the procedure \textsc{SetV}, constructs a high dimensional matrix and finds its strict PSNE which maps to $\calX_\star\times \calY_\star$. Finally, as we show in \Cref{sec:relax-k1k2}, it is straightforward to relax the assumptions on knowing $k_1$ and $k_2$ by sequentially trying all possible combinations, starting from $k_1=1$ and $k_2=1$, and stopping when we identify $\calX_\star \times \calY_\star$. Throughout this process, we use the procedure \psnealg/ to either find the strict PSNE or compute the set $\calV$ using appropriate query oracles, ensuring that the query complexity remains nearly $n \cdot \text{poly}(k)$.

\subsubsection{Finding the value of the game candidates with the knowledge of $k_1$ and $k_2$}\label{sec:setv}
This section involves three key steps: defining a large matrix, defining the query oracle, and identifying a set $\calV$ of value of the game candidates. 

We begin by assuming that we know the values $k_1 := |\bigcup\limits_{x\in \calX_\star}\supp(x)|$ and $k_2 := |\bigcup\limits_{y\in \calY_\star}\supp(y)|$. Let $A\in\bbR^{n\times n}$ be the input matrix. We construct a matrix $A^{(k_1,k_2)}\in \bbR^{\binom{n}{k_1} \times \binom{n}{k_2}}$ as follows. First, we consider two bijections $f_1:[\binom{n}{k_1}]\rightarrow \binom{[n]}{k_1}$ and $f_2:[\binom{n}{k_2}]\rightarrow \binom{[n]}{k_2}$. Next, we consider indices $i\in [\binom{n}{k_1}]$ and $j\in [\binom{n}{k_2}]$. Let $M$ be the submatrix of $A$ with row indices $f_1(i)$ and column indices $f_2(j)$. We now define $A_{i,j}^{(k_1,k_2)}:=V_M^\star$. Now we prove the following lemma.
\begin{lemma}\label{lem:high-dim-psne-0}
    Let $i_\star:=f_1^{-1}(\bigcup\limits_{x\in \calX_\star}\supp(x))$ and $j_\star:=f_2^{-1}(\bigcup\limits_{y\in \calY_\star}\supp(y))$. Then $(i_\star,j_\star)$ is a PSNE of $A^{(k_1,k_2)}$ and $A^{(k_1,k_2)}_{\istar,\jstar}=V_A^\star$.
\end{lemma}
\begin{proof}
For the sake of simplicity in presentation, let $\widehat A$ denote the matrix $A^{(k_1,k_2)}$. Recall that $\simplex_{\calI}$ represents the simplex over a set of indices $\calI$. Let $(x_\star,y_\star)$ be a Nash equilibrium of $A$ such that $\supp(x_\star)=\bigcup\limits_{x\in \calX_\star}\supp(x)$ and $\supp(y_\star)=\bigcup\limits_{y\in \calY_\star}\supp(y)$.
         Such an equilibrium exists due to the fact that convex combination of multiple Nash equilibria is also a Nash equilibrium. Let $\widehat x\in \simplex_{f_1(i_\star)}, \widehat y\in\simplex_{f_2(j_\star)}$ be strategies such that for all $i\in f_1(i_\star)$ we have $\widehat x_i=(x_\star)_i$ and for all $j\in f_2(j_\star)$ we have $\widehat y_j=(y_\star)_j$. 

        First we claim that $\widehat A_{i_\star,j_\star}=V_A^\star$. Consider the submatrix $M$ of $A$ with row indices $f_1(i_\star)$ and column indices $f_2(j_\star)$. As $(x_\star,y_\star)$ is a Nash equilibrium of $A$, and $\supp(x_\star)=f_1(i_\star)$ and $\supp(y_\star)=f_2(j_\star)$, we have $\langle x,M \widehat y\rangle=V_A^\star$ for all $x\in \simplex_{f_1(i_\star)}$ and $\langle \widehat x,M y\rangle=V_A^\star$ for all $y\in \simplex_{f_2(j_\star)}$. Hence $(\widehat x,\widehat y)$ is the Nash equilibrium of $M$ which implies that $\widehat A_{i_\star,j_\star}=V_M^\star=\langle \widehat x,M\widehat y\rangle=V_A^\star$.

        Next consider an index $i\in [\binom{n}{k_1}]$. Now we claim that $\widehat A_{i,j_\star}\leq V_A^\star$.  Consider the submatrix $M$ of $A$ with row indices $f_1(i)$ and column indices $f_2(j_\star)$. Now we have the following:
\[
\widehat A_{i,j_\star}=V_{M}^\star=\min_{y\in\simplex_{f_2(j_\star)}}\max_{i\in f_1(i)}\langle e_i,My\rangle\leq \max_{i\in f_1(i)}\langle e_i,M \widehat y\rangle\leq V_A^\star.
\]
The last inequality follows from the definition of $\widehat y$ and the fact that $\max_{i\in[n]}\langle e_i,Ay_\star\rangle=V_A^\star$. Similarly we can show that for any $j\in [\binom{n}{k_2}]$, we have $\widehat A_{i_\star,j}\geq V_A^\star$. Hence, $(\istar,\jstar)$ is a PSNE.
\end{proof}

Now we describe the query oracle to observe the rows and columns of $A^{(k_1,k_2)}$.

\textbf{Query oracle $\calQ$ to observe row $\hat i$ and column $\hat j$ of $A^{(k_1,k_2)}$}:
Recall that $A^{(k_1,k_2)}_{i,j}$ only depends on the entries in the submatrix $M$ with row indices $f_1(i)$ and column indices $f_2(j)$.
Hence, if we want to observe the values in the set $\{A^{(k_1,k_2)}_{\hat i,j}:j\in \{1,2,\ldots,\binom{n}{k_2}\}\}$, then it suffices to query the entries of $A$ from the set $S_1=\{(i,j):i\in f_{1}(\hat i), j\in[n]\}$.
Since $|f_1(\hat i)|=k_1$, we have that $|S_1|= k_1 \cdot n$.
Analogously, if we want to observe the values in the set $\{A^{(k_1,k_2)}_{i,\hat j}:i\in \{1,2,\ldots,\binom{n}{k_1}\}\}$, then it suffices to query the entries of $A$ from the set $S_2=\{(i,j):i\in [n], j\in f_{2}(\hat j)\}$. 
Since $|f_2(\hat j)|=k_2$, we have $|S_2|= k_2 \cdot n$.

Now we describe $\textsc{SetV}$, a randomized algorithm that aims to find a set $\calV$ that includes $V_A^\star$ .  

\textbf{Description of the randomized procedure } \textsc{SetV}$(A,\delta,k_1,k_2)$:
Let $A\in\bbR^{n\times n}$ be the input matrix. Recall the construction of the matrix $A^{(k_1,k_2)}\in \bbR^{\binom{n}{k_1} \times \binom{n}{k_2}}$ from above. We now call the procedure \psnealg/($A^{(k_1,k_2)},\delta$), where we use the query oracle $\calQ$ to observe a row or column of $A^{(k_1,k_2)}$. Let $\calV$ be the set of values returned by \psnealg/($A^{(k_1,k_2)},\delta$). We return the set $\calV$.

Now, we prove the following lemma regarding the procedure \textsc{SetV}.
\begin{lemma}\label{lem:setv}
Consider a matrix game $A\in\bbR^{n\times n}$. Now we have the following:
\begin{itemize}
    \item The procedure \textsc{SetV}$(A,\delta,k_1,k_2)$ queries at most $O(nk^3\cdot \log n\cdot \log(\frac{n}{\delta}))$ entries of $A$. 
    \item With probability at least $1-\delta$, the procedure \textsc{SetV}$(A,\delta,k_1,k_2)$ returns a set $\calV$ of distinct real numbers such that $V_A^\star\in\calV$.
\end{itemize}
\end{lemma}
\begin{proof}
    Due to \Cref{lem:psne-query}, the number of calls made by \psnealg/($A^{(k_1,k_2)},\delta$) to the query oracle $\calQ$ is at most $O(\log(n^k)\cdot\log(\frac{n^k}{\delta}))\leq O(k^2\cdot \log(n)\cdot\log(\frac{n}{\delta}))$. Each call to the query oracle $\calQ$ takes $O(nk)$ queries from the input matrix $A$. Hence, the procedure \psnealg/($A^{(k_1,k_2)},\delta$) queries at most $O(nk^3\cdot \log n\cdot \log(\frac{n}{\delta}))$ entries of $A$ and returns a set $\calV$. Due to \Cref{lem:non-unique-psne,lem:high-dim-psne-0}, with probability at least $1-\delta$, $V_A^\star\in\calV$. 
\end{proof}

\subsubsection{Algorithm for finding all Nash equilibria with the knowledge of $k_1$ and $k_2$}\label{sec:k1k2-known}
This section involves three key steps: modifying the large matrix from the previous section, defining the query oracle, and identifying the set of all Nash equilibria.

We begin by assuming that we know the values $k_1 := |\bigcup\limits_{x\in \calX_\star}\supp(x)|$ and $k_2 := |\bigcup\limits_{y\in \calY_\star}\supp(y)|$.
Let $A\in\bbR^{n\times n}$ be the input matrix. Let $\calV$ be a set of distinct real numbers. Recall that $\simplex_{\calI}$ represents the simplex over a set of indices $\calI$. We construct a matrix $A^{(k_1,k_2,\calV)}\in \bbR^{\binom{n}{k_1} \times \binom{n}{k_2}}$ as follows. First, we consider two bijections $f_1:[\binom{n}{k_1}]\rightarrow \binom{[n]}{k_1}$ and $f_2:[\binom{n}{k_2}]\rightarrow \binom{[n]}{k_2}$. Next, we consider indices $i\in [\binom{n}{k_1}]$ and $j\in [\binom{n}{k_2}]$. Let $M$ be the submatrix of $A$ with row indices $f_1(i)$ and column indices $f_2(j)$. Let $\calX_M:=\arg\max_{x\in\simplex_{f_1(i)}}\min_{y\in\simplex_{f_2(j)}}\langle x,My\rangle$ and $\calY_M:=\arg\min_{y\in\simplex_{f_2(j)}}\max_{x\in\simplex_{f_1(j)}}\langle x,My\rangle$. Let $k_{1,M}:=|\bigcup\limits_{x\in \calX_M}\supp(x)|$ and $k_{2,M}:=|\bigcup\limits_{y\in \calY_M}\supp(y)|$. If $V_M^\star\notin \calV$ then we define $A_{i,j}^{(k_1,k_2,\calV)}:=V_M^\star$. If $V_M^\star\in \calV$ then we define $A_{i,j}^{(k_1,k_2,\calV)}$ as follows:
\begin{itemize}
    \item If $k_{1,M}=k_1$ and $k_{2,M}=k_2$ then $A_{i,j}^{(k_1,k_2,\calV)}=V_M^\star$
    \item If $k_{1,M}\neq k_1$ and $k_{2,M}\neq k_2$ then $A_{i,j}^{(k_1,k_2,\calV)}=V_M^\star$
    \item If $k_{1,M}= k_1$ and $k_{2,M}\neq k_2$ then $A_{i,j}^{(k_1,k_2,\calV)}=V_M^\star+\frac{\varepsilon}{4}$
    \item If $k_{1,M}\neq k_1$ and $k_{2,M}=k_2$ then $A_{i,j}^{(k_1,k_2,\calV)}=V_M^\star-\frac{\varepsilon}{4}$
\end{itemize}
where $\varepsilon=\min\limits_{v_1,v_2\in\calV:v_1\neq v_2}|v_1-v_2|$. If $|\calV|=1$, then $\varepsilon=1$. Now we prove the following lemma.
\begin{lemma}\label{lem:high-dim-psne-1}
    Let $i_\star:=f_1^{-1}(\bigcup\limits_{x\in \calX_\star}\supp(x))$ and $j_\star:=f_2^{-1}(\bigcup\limits_{y\in \calY_\star}\supp(y))$. If $V_A^\star\in\calV$, then $(i_\star,j_\star)$ is the only strict PSNE of $A^{(k_1,k_2,\calV)}$.
\end{lemma}
\begin{proof} Let us assume that $V_A^\star\in\calV$. For the sake of simplicity in presentation, let $\widehat A$ denote the matrix $A^{(k_1,k_2,\calV)}$.
    Now due to \Cref{Unique-PSNE} it suffices to show that 
        $\widehat A_{i,j_\star}<\widehat A_{i_\star,j_\star}<\widehat A_{i_\star,j}$ for all $i\in [\binom{n}{k_1}]\setminus\{i_\star\}$ and $j\in [\binom{n}{k_2}] \setminus\{j_\star\}$.

        Let $(x_\star,y_\star)$ be a Nash equilibrium of $A$ such that the following holds:
        \begin{itemize}
            \item $\supp(x_\star)=\bigcup\limits_{x\in \calX_\star}\supp(x)$ and $\supp(y_\star)=\bigcup\limits_{y\in \calY_\star}\supp(y)$
            \item For all $i\in \supp(x_\star)$, we have $\langle e_i,Ay_\star\rangle=V_A^\star$, and for all $i\notin \supp(x_\star)$, we have $\langle e_i,Ay_\star\rangle<V_A^\star$
             \item For all $j\in \supp(y_\star)$, we have $\langle x_\star,Ae_j\rangle=V_A^\star$, and for all $j\notin \supp(y_\star)$, we have $\langle x_\star,Ae_j\rangle>V_A^\star$
        \end{itemize}
         Such an equilibrium exists due to \Cref{technical-multiple-NE} and the fact that convex combination of multiple Nash equilibria is also a Nash equilibrium. Let $\widehat x\in \simplex_{f_1(i_\star)}, \widehat y\in\simplex_{f_2(j_\star)}$ be strategies such that for all $i\in f_1(i_\star)$ we have $\widehat x_i=(x_\star)_i$ and for all $j\in f_2(j_\star)$ we have $\widehat y_j=(y_\star)_j$. 

        First we claim that $\widehat A_{i_\star,j_\star}=V_A^\star$. Consider the submatrix $M$ of $A$ with row indices $f_1(i_\star)$ and column indices $f_2(j_\star)$. As $\supp(x_\star)=f_1(\istar)$ and $\supp(y_\star)=f_2(\jstar)$, we have $\langle x,M \widehat y\rangle=V_A^\star$ for all $x\in \simplex_{f_1(i_\star)}$ and $\langle \widehat x,M y\rangle=V_A^\star$ for all $y\in \simplex_{f_2(j_\star)}$. Hence $(\widehat x,\widehat y)$ is the Nash equilibrium of $M$ which implies that $V_M^\star=\langle \widehat x,M\widehat y\rangle=V_A^\star$. As $|\supp(\widehat x)|=k_1$ and $|\supp(\widehat y)|=k_2$ we have $\widehat A_{i_\star,j_\star}=V_A^\star$.

        Next consider an index $i\in [\binom{n}{k_1}]$ such that $f_1(i)\cap f_1(i_\star)=\emptyset$. Now we claim that $\widehat A_{i,j_\star}<V_A^\star$.  Consider the submatrix $M$ of $A$ with row indices $f_1(i)$ and column indices $f_2(j_\star)$. Now we have the following:
\[
V_{M}^\star=\min_{y\in\simplex_{f_2(j_\star)}}\max_{i\in f_1(i)}\langle e_i,My\rangle\leq \max_{i\in f_1(i)}\langle e_i,M \widehat y\rangle<V_A^\star.
\]
The last inequality follows from our choice of $(x_\star,y_\star)$ above and $f_1(i)\cap f_1(i_\star)=\emptyset$. If $V_M^\star\notin \calV$, then $\widehat A_{i,j_\star}=V_M^\star<V_A^\star$. If $V_M^\star\in\calV$, then $\widehat A_{i,j_\star}\leq V_M^\star+\frac{\varepsilon}{4}<V_A^\star$. The last inequality follows from the definition $\varepsilon$. Analogously, we can show that for any $j\in [\binom{n}{k_2}]$ such that $f_2(j)\cap f_2(j_\star)=\emptyset$, we have $\widehat A_{i_\star,j}>V_A^\star$.

Next consider an index $i\in [\binom{n}{k_1}]\setminus\{i_\star\}$ such that $f_1(i)\cap f_1(i_\star)\neq\emptyset$. Now we claim that $\widehat A_{i,j_\star}<V_A^\star$.  Consider the submatrix $M$ of $A$ with row indices $f_1(i)$ and column indices $f_2(j_\star)$. Now we have the following:
\[
V_{M}^\star=\min_{y\in\simplex_{f_2(j_\star)}}\max_{i\in f_1(i)}\langle e_i,My\rangle\leq \max_{i\in f_1(i)}\langle e_i,M \widehat y\rangle= V_A^\star.
\]
Recall the definition of $\calX_M, \calY_M, k_{1,M}$ and $k_{2,M}$. If $V_M^\star<V_A^\star$ and $V_M^\star\notin\calV$, then we have $\widehat A_{i,j_\star}=V_M^\star<V_A^\star$. If $V_M^\star<V_A^\star$ and $V_M^\star\in\calV$, then we have $\widehat A_{i,j_\star}\leq V_M^\star+\frac{\varepsilon}{4}<V_A^\star$. Let us consider the case when $V_M^\star=V_A^\star$. In this case, we have $\widehat y\in \calY_{M}$ which implies that $k_{2,M}=k_2$. Consider $x\in \calX_{M}$. For any $i\notin f_1(i)\cap f_1(i_\star)$ we have $\langle e_i,M\widehat y\rangle<V_A^\star$. This implies that $\supp(x)\subseteq f_1(i)\cap f_1(i_\star)$. As $x$ was chosen arbitrarily, we have $\bigcup\limits_{x\in \calX_M}\supp(x)\subseteq f_1(i)\cap f_1(i_\star)$. As $f_1(i)\neq f_1(i_\star)$ this implies $|f_1(i)\cap f_1(i_\star)|<k_1$. This implies that $k_{1,M}\neq k_1$. Hence due to our construction, we have $\widehat A_{i,j_\star}=V_A^\star-\frac{\varepsilon}{4}<V_A^\star$. Analogously, we can show that for any $j\in [\binom{n}{k_2}]\setminus\{j_\star\}$ such that $f_2(j)\cap f_2(j_\star)\neq\emptyset$, we have $\widehat A_{i_\star,j}>V_A^\star$.

Combining all the cases above, we have  $\widehat A_{i,j_\star}<\widehat A_{i_\star,j_\star}<\widehat A_{i_\star,j}$ for all $i\in [\binom{n}{k_1}]\setminus\{i_\star\}$ and $j\in [\binom{n}{k_2}] \setminus\{j_\star\}$.

\end{proof}

Now we describe the query oracle to observe the rows and columns of $A^{(k_1,k_2,\calV)}$.

\textbf{Query oracle $\calQ$ to observe row $\hat i$ and column $\hat j$ of $A^{(k_1,k_2,\calV)}$}:
Recall that $A^{(k_1,k_2,\calV)}_{i,j}$ only depends on the set $\calV$ and the entries in the submatrix $M$ with row indices $f_1(i)$ and column indices $f_2(j)$.
Hence, if we want to observe the values in the set $\{A^{(k_1,k_2,\calV)}_{\hat i,j}:j\in \{1,2,\ldots,\binom{n}{k_2}\}\}$, then it suffices to query the entries of $A$ from the set $S_1=\{(i,j):i\in f_{1}(\hat i), j\in[n]\}$.
Since $|f_1(\hat i)|=k_1$, we have that $|S_1|= k_1 \cdot n$.
Analogously, if we want to observe the values in the set $\{A^{(k_1,k_2,\calV)}_{i,\hat j}:i\in \{1,2,\ldots,\binom{n}{k_1}\}\}$, then it suffices to query the entries of $A$ from the set $S_2=\{(i,j):i\in [n], j\in f_{2}(\hat j)\}$. 
Since $|f_2(\hat j)|=k_2$, we have $|S_2|= k_2 \cdot n$.

Now we describe $\textsc{Nash}$, a randomized algorithm that finds $\calX_\star$ and $\calY_\star$ when provided with the values $k_1$ and $k_2$. 

\textbf{Description of the randomized procedure } \textsc{Nash}$(A,\delta,k_1,k_2)$:
Let $A\in\bbR^{n\times n}$ be the input matrix. First, we call the procedure \textsc{SetV}$(A,\delta/2,k_1,k_2)$. Let $\calV$ be the set returned by this procedure. Recall the construction of the matrix $A^{(k_1,k_2,\calV)}\in \bbR^{\binom{n}{k_1} \times \binom{n}{k_2}}$ from above. We now call the procedure \psnealg/($A^{(k_1,k_2,\calV)},\delta/2$), where we use the query oracle $\calQ$ to observe a row or column of $A^{(k_1,k_2,\calV)}$.

Let $(\hat i, \hat j)$ be the pair returned by the procedure \psnealg/$(A^{(k_1,k_2,\calV)},\delta/2)$. Let $M_1$ be the submatrix of $A$ with row indices $f_1(\hat i)$ and column indices $[n]$, and $M_2$ be the submatrix of $A$ with row indices $[n]$ and column indices $f_2(\hat j)$. Let $\calX_{M_1}:=\arg\max_{x\in\simplex_{f_1(\hat i)}}\min_{y\in\simplex_{n}}\langle x,M_1y\rangle$ and $\calY_{ M_2}:=\arg\min_{y\in\simplex_{f_2(\hat j)}}\max_{x\in\simplex_{n}}\langle x,M_2y\rangle$. For any vector $v$ and the set of indices $\calS=\{i_1,\ldots, i_{\ell}\}$, let $v_\calS$ denote the vector $(v_{i_1},\ldots,v_{i_{\ell}})$. Let $\calX_A:=\{x\in\simplex_n:x_{f_1(\hat i)}\in \calX_{M_1} \}$ and $\calY_A:=\{y\in\simplex_m:y_{f_2(\hat j)}\in \calY_{M_2} \}$. We return $\calX_A\times \calY_A$ as the set of all Nash equilibria if the following conditions hold:
\begin{enumerate}
    \item[(C1)] $V_{M_1}^\star=V_{M_2}^\star$. 
    \item[(C2)] $\bigcap\limits_{x\in\calX_A}I_x'=f_2(\hat j)$ where $I_x'=\{j\in[n]: V_{M_1}^\star=\langle x, A e_j\rangle\}$.
    \item[(C3)] $\bigcap\limits_{y\in\calY_A}I_y'=f_1(\hat i)$ where $I_y'=\{i\in[n]: V_{M_2}^\star=\langle e_i, A y\rangle\}$.
\end{enumerate}

The above conditions can be checked by querying all the entries of $A$ in the set $\{(i,j):i\in f_1(\hat i),j\in[n]\}\cup\{(i,j):i\in [n],j\in f_2(\hat j)\}$. If any one of the conditions above does not hold, we return "failure" instead.

Now we state the guarantees of our randomized procedure.
\begin{lemma}\label{lem:alg-multiple-ne}
Consider a matrix game $A\in\bbR^{n\times n}$. Now we have the following:
\begin{itemize}
    \item The procedure \textsc{Nash}$(A,\delta,k_1,k_2)$ queries at most $O(nk^3\cdot \log n\cdot \log(\frac{n}{\delta}))$ entries of $A$. 
    \item With probability at least $1-\delta$, the procedure \textsc{Nash}$(A,\delta,k_1,k_2)$ returns $\calX_\star\times \calY_\star$. 
\end{itemize}

\end{lemma}
\begin{proof}
Due to \Cref{lem:setv}, \textsc{SetV}$(A,\delta/2,k_1,k_2)$ queries at most $O(nk^3\cdot \log n\cdot \log(\frac{n}{\delta}))$ entries of $A$. Next due to \Cref{lem:psne-query}, the number of calls made by \psnealg/($A^{(k_1,k_2,\calV)},\delta/2$) to the query oracle $\calQ$ is at most $O(\log(n^k)\cdot\log(\frac{n^k}{\delta}))\leq O(k^2\cdot \log(n)\cdot\log(\frac{n}{\delta}))$. Each call to the query oracle $\calQ$ takes $O(nk)$ queries from the input matrix $A$. Moreover, to check the conditions (C1), (C2) and (C3) it takes $O(nk)$ queries. Hence, the procedure \textsc{Nash}$(A,\delta,k_1,k_2)$ queries at most $O(nk^3\cdot \log n\cdot \log(\frac{n}{\delta}))$ entries of $A$.   

Due \Cref{lem:setv}, with probability at least $1-\frac{\delta}{2}$, \textsc{SetV}$(A,\delta/2,k_1,k_2)$ returns a set $\calV$ such that $V_A^\star\in\calV$.
Recall the construction of the sets $\calX_A$ and $\calY_A$. Due to \Cref{lem:high-dim-psne-1} and \Cref{lem:psne-correct}, the procedure \psnealg/$(A^{(k_1,k_2,\calV)},\delta/2)$ returns the pair $(\istar, \jstar)$ with probability at least $1-\frac{\delta}{2}$ if $V_A^\star\in\calV$, where $i_\star:=f_1^{-1}(\bigcup\limits_{x\in \calX_\star}\supp(x))$ and $j_\star:=f_2^{-1}(\bigcup\limits_{y\in \calY_\star}\supp(y))$. The previous facts along with  \Cref{technical:correctness} imply that with probability at least $1-\delta$, we have $\calX_A=\calX_\star$ and $\calY_A=\calY_\star$. As the set $\calX_\star\times\calY_\star$ satisfies the conditions (C1), (C2) and (C3) due to  \Cref{technical:correctness}, with probability at least $1-\delta$, the procedure \textsc{Nash}$(A,\delta,k_1,k_2)$ returns $\calX_\star\times \calY_\star$.

\end{proof}

\subsubsection{Relaxing the assumption on the knowledge of $k_1$ and $k_2$.}\label{sec:relax-k1k2}

It is straightforward to relax the assumption on the knowledge of $k_1$ and $k_2$. 
Beginning with $s=1$, we call the procedure \textsc{Nash}($A,\delta,s_1,s_2$) on all possible pairs $(s_1,s_2)$ such that $\max\{s_1,s_2\}=s$. Note that the number of such pairs is $2s-1$. If for any one of the pairs $(s_1,s_2)$ the procedure \textsc{Nash}($A,\delta,s_1,s_2$) returns a set $\calX_A\times \calY_A$, we return it as the set of all Nash equilibria and terminate the process. If for all the pairs $(s_1,s_2)$ the procedure \textsc{Nash}($A,\delta,s_1,s_2$) returns "failure" then we increment $s$ by one and repeat the process. With probability $1-\delta$, this process terminates at $s=k$ and returns the set $\calX_\star\times\calY_\star$. This is because the set $\calX_\star\times\calY_\star$ is the only set that satisfies the conditions (C1), (C2) and (C3) due to \Cref{technical:correctness}.  Hence, with probability $1-\delta$, the total number of queries taken from $A$ is at most $(\sum_{s=1}^k2s-1)\cdot O( nk^3\cdot \log n\cdot \log(\frac{n}{\delta}))=O(nk^5\cdot \log n\cdot \log(\frac{n}{\delta}))$.

\subsection{Lower bound}\label{sec:lower-bound}
In this section, we prove the following lower bound result on the query complexity.

\lowerr*

\begin{proof}


For each $\ell\in[n]$, let $\calI_{\ell}$ be the set of all $n\times n$ matrices such that the following holds:
\begin{enumerate}
    \item all entries $(i,j)$ have a value two for $i\in [n]$ and $j\in[n]\setminus [k]$
    \item for each $i\in [n]$ such that $i\notin\{1\}\cup\{\ell\}$, there exists exactly one index $j\in[k]$ such that the entry $(i,j)$ has a value of zero.
    \item rest of the entries of the matrix have a value of one.
\end{enumerate}
Fix an index $\ell\in [n]$. For any matrix in $\calI_{\ell}$, it can be easily shown using \Cref{technical-multiple-NE} that $\calS_\ell:=\{(e_i,e_j):i\in\{1\}\cup\{\ell\},j\in [k]\}$ is the set of all pure strategy nash equilibria, and that their convex combination is the set of all Nash equilibria. Also observe that for two distinct indices $\ell_1,\ell_2\in[n]$, we have $\calS_{\ell_1}\neq \calS_{\ell_2}$.

Let $\calD$ be a uniform distribution over the set of matrices in $\calI_1$. Observe that for any matrix $A\in \calI_1$, we have $\max\{k_1,k_2\}=k$ and $\min\{k_1,k_2\}=1$.

Let us assume that the randomized algorithm $\calA$ is a distribution over a set $\calT$ of deterministic algorithms.
Fix a deterministic instance $\alg\in \calT$. Let $\tau(\alg,A)$ denote the number of queries taken by $\alg$ from an input matrix $A$. We now aim to show that $\mathbb{E}_{A\sim \calD}[\tau(\alg,A)]\geq \frac{(n-1)(k+1)}{2}$. For the sake of simplicity in presentation, let us assume that $\alg$ does not query an entry more than once.

Consider an input matrix $A\in\calI_1$. As $\alg$ correctly computes the set of all PSNEs of the input matrix, $\alg$ has to query all those entries of $A$ that have a value of $0$. For the sake of contradiction, let us assume that $\alg$ does not query an entry $(i',j')$ that has a value of $0$. One can create a matrix $B\in \calI_{i'}$ such that for all $(i,j)\neq (i',j')$ we have $B_{i,j}=A_{i,j}$ and $B_{i',j'}=1$. As $\alg$ terminates without querying the entry $(i',j')$ it cannot distinguish between $A$ and $B$ and hence outputs the wrong set of PSNEs for at least one of the matrices which is contradiction to our assumption. Hence, all the $n-1$ zero-valued entries are always queried for all matrices in $\calI_1$.

We continue to focus on the input matrices $A$ in $\calI_1$. Fix $i\in [n]\setminus\{1\}$. Let $X_{i,A}$ denote the number of entries of the set $\{(i,j):j\in[k]\}$ that $\alg$ queries before querying the zero-valued entry in row $i$ of an input matrix $A$. Let $\calF_i$ be the family of functions $f:[n]\setminus\{1,i\}\rightarrow [k]$. Fix $f\in \calF_i$. Let $\calI_{1,f}$ be the largest subset of $\calI_1$ such that for all $A\in\calI_{1,f}$ we have $A_{\hat i,f(\hat i)}=0$ for all $\hat i\in [n]\setminus\{1,i\}$. Let $(i,j_1)$ be the first entry to be queried from the set $\{(i,j):j\in[k]\}$  
 by $\alg$ if input matrix $A$ is chosen from the set $\calI_{1,f}$. Recursively for $\ell\geq 2$, let $(i,j_\ell)$ be the $\ell$-th entry to be queried from the set $\{(i,j):j\in[k]\}$  
 by $\alg$ if input matrix $A$ is chosen from the set $\calI_{1,f}$ and the first $\ell-1$ entries queried from the set $\{(i,j):j\in[k]\}$ have a value of one. Note that the sequence $j_1,\ldots,j_k$ is fixed as $\alg$ behaves identically on all the matrices in $\calI_{1,f}$ until it queries the entry in row $i$ that has a value of $0$.
Fix $m\in \{0,\ldots,k-1\}$. Now observe that $\mathbb{P}_{A\sim \calD}[X_{i,A}=m| A\in\calI_{1,f}]=\mathbb{P}_{A\sim \calD}[A_{i,j_{m+1}}=0| A\in\calI_{1,f}]=\frac{1}{k}$ as $|\calI_{1,f}|=k$ and there is exactly one matrix $A\in \calI_{1,f}$ such that $A_{i,j_{m+1}}=0$. Now we have the following:
\begin{align*}
    \mathbb{P}_{A\sim \calD}[X_{i,A}=m]&=\sum_{f\in\calF_i}\mathbb{P}_{A\sim \calD}[X_{i,A}=m| A\in\calI_{1,f}]\cdot \mathbb{P}_{A\sim \calD}[A\in\calI_{1,f}]\\
    &=\frac{1}{k}\cdot\sum_{f\in\calF_i}\mathbb{P}_{A\sim \calD}[A\in\calI_{1,f}]\\
    &=\frac{1}{k} \tag{as $\{\calI_{1,f}\}_{f\in\calF_i}$ is a partition of $\calI_1$}
\end{align*}

 Hence, $\mathbb{E}_{A\sim \calD}[X_{i,A}]=\sum_{m=0}^{k-1}\frac{m}{k}=\frac{k-1}{2}$. Now observe that $\mathbb{E}_{A\sim \calD}[\tau(\alg,A)]\geq (n-1)+\sum_{i=2}^{n}\mathbb{E}_{A\sim \calD}[X_{i,A}]=\frac{(n-1)(k+1)}{2}$. Now due to Yao's lemma, we have that
\[
\max_{A\in\calI_1}\E[\tau(A)] \geq \min_{\alg\in\calT}\E_{A\sim\calD}[\tau(\alg,A)] \geq\frac{(n-1)(k+1)}{2}.
\] 
\end{proof}
\section{{Conclusion and Open Questions}}
This paper characterizes the query complexity of finding the set of all Nash equilibria, $\calX_\star\times\calY_\star$, in two-player zero-sum matrix games. In Theorem \ref{thm:algo-result}, we demonstrated that with probability at least $1-\delta$, the set $\calX_\star\times\calY_\star$ can be found by querying at most $O(nk^5\cdot\polylog(\frac{n}{\delta}))$ entries. Additionally, in Theorem \ref{thm:lower-multiple-ne}, we established a lower bound of $\Omega(nk)$. Bridging the gap between this upper and lower bound remains an important open question.

Furthermore, note that our algorithm does not run in polynomial time. Therefore, designing a polynomial-time randomized algorithm that requires $o(n^2)$ queries and finds a Nash equilibrium presents another compelling open problem. We conjecture that such an algorithm would need to be strongly polynomial.

Finally, it is worth recalling that there exists a deterministic algorithm with a query complexity of $O(n)$ to find the strict PSNE. The question of whether a deterministic algorithm can be developed to find the set of all Nash equilibria with $o(n^2)$ queries is yet another important open question.

\section*{ACKNOWLEDGEMENTS}
This work was supported in part by NSF TRIPODS CCF Award \#2023166 and a Northrop Grumman University Research Award. We thank Prof. Rahul Savani for his feedback on the paper.
\bibliographystyle{plainnat}
\bibliography{refs.bib}

\begin{thebibliography}{22}
\providecommand{\natexlab}[1]{#1}
\providecommand{\url}[1]{\texttt{#1}}
\expandafter\ifx\csname urlstyle\endcsname\relax
  \providecommand{\doi}[1]{doi: #1}\else
  \providecommand{\doi}{doi: \begingroup \urlstyle{rm}\Url}\fi

\bibitem[Adler(2013)]{adler2013equivalence}
Ilan Adler.
\newblock The equivalence of linear programs and zero-sum games.
\newblock \emph{International Journal of Game Theory}, 42\penalty0
  (1):\penalty0 165, 2013.

\bibitem[Auger et~al.(2014)Auger, Liu, Ruette, St-Pierre, and
  Teytaud]{auger2014sparse}
David Auger, Jialin Liu, Sylvie Ruette, David~Lupien St-Pierre, and Olivier
  Teytaud.
\newblock Sparse binary zero-sum games.
\newblock In \emph{ACML}, 2014.

\bibitem[Babichenko(2016)]{babichenko2016query}
Yakov Babichenko.
\newblock Query complexity of approximate nash equilibria.
\newblock \emph{Journal of the ACM (JACM)}, 63\penalty0 (4):\penalty0 1--24,
  2016.

\bibitem[Babichenko(2020)]{babichenko2020informational}
Yakov Babichenko.
\newblock Informational bounds on equilibria (a survey).
\newblock \emph{ACM SIGecom Exchanges}, 17\penalty0 (2):\penalty0 25--45, 2020.

\bibitem[Bienstock et~al.(1991)Bienstock, Chung, Fredman, Sch{\"a}ffer, Shor,
  and Suri]{bienstock1991note}
Daniel Bienstock, Fan Chung, Michael~L Fredman, Alejandro~A Sch{\"a}ffer,
  Peter~W Shor, and Subhash Suri.
\newblock A note on finding a strict saddlepoint.
\newblock \emph{The American mathematical monthly}, 98\penalty0 (5):\penalty0
  418--419, 1991.

\bibitem[Bohnenblust et~al.(1948)Bohnenblust, Girshick, Snow, Dresher,
  Helmer-Hirschberg, McKinsey, Shapley, and
  Harris]{bohnenblust1948mathematical}
HF~Bohnenblust, MA~Girshick, RN~Snow, Melvin Dresher, Olaf Helmer-Hirschberg,
  JCC McKinsey, Lloyd~S Shapley, and Theodore~Edward Harris.
\newblock Mathematical theory of zero-sum two-person games with a finite number
  or a continuum of strategies.
\newblock 1948.

\bibitem[Bohnenblust et~al.(1950)Bohnenblust, Karlin, and
  Shapley]{bohnenblust1950solutions}
HF~Bohnenblust, S~Karlin, and LS~Shapley.
\newblock Solutions of discrete, two-person games.
\newblock \emph{Contributions to the Theory of Games}, 1:\penalty0 51--72,
  1950.

\bibitem[Brooks and Reny(2021)]{brooks2021canonical}
Benjamin Brooks and Philip~J Reny.
\newblock A canonical game--nearly 75 years in the making--showing the
  equivalence of matrix games and linear programming.
\newblock \emph{Available at SSRN 3851583}, 2021.

\bibitem[Chen et~al.(2009)Chen, Deng, and Teng]{chen2009settling}
Xi~Chen, Xiaotie Deng, and Shang-Hua Teng.
\newblock Settling the complexity of computing two-player nash equilibria.
\newblock \emph{Journal of the ACM (JACM)}, 56\penalty0 (3):\penalty0 1--57,
  2009.

\bibitem[Dallant et~al.(2024{\natexlab{a}})Dallant, Haagensen, Jacob, Kozma,
  and Wild]{dallant2024finding}
Justin Dallant, Frederik Haagensen, Riko Jacob, L{\'a}szl{\'o} Kozma, and
  Sebastian Wild.
\newblock Finding the saddlepoint faster than sorting.
\newblock In \emph{2024 Symposium on Simplicity in Algorithms (SOSA)}, pages
  168--178. SIAM, 2024{\natexlab{a}}.

\bibitem[Dallant et~al.(2024{\natexlab{b}})Dallant, Haagensen, Jacob, Kozma,
  and Wild]{dallant2024optimal}
Justin Dallant, Frederik Haagensen, Riko Jacob, L{\'a}szl{\'o} Kozma, and
  Sebastian Wild.
\newblock An optimal randomized algorithm for finding the saddlepoint.
\newblock \emph{arXiv preprint arXiv:2401.06512}, 2024{\natexlab{b}}.

\bibitem[Dantzig(1951)]{dantzig1951proof}
George~B Dantzig.
\newblock A proof of the equivalence of the programming problem and the game
  problem.
\newblock \emph{Activity analysis of production and allocation}, 13, 1951.

\bibitem[Daskalakis et~al.(2007)Daskalakis, Mehta, and
  Papadimitriou]{daskalakis2007progress}
Constantinos Daskalakis, Aranyak Mehta, and Christos Papadimitriou.
\newblock Progress in approximate nash equilibria.
\newblock In \emph{Proceedings of the 8th ACM Conference on Electronic
  Commerce}, pages 355--358, 2007.

\bibitem[Daskalakis et~al.(2009)Daskalakis, Mehta, and
  Papadimitriou]{daskalakis2009note}
Constantinos Daskalakis, Aranyak Mehta, and Christos Papadimitriou.
\newblock A note on approximate nash equilibria.
\newblock \emph{Theoretical Computer Science}, 410\penalty0 (17):\penalty0
  1581--1588, 2009.

\bibitem[Deligkas et~al.(2022)Deligkas, Fasoulakis, and
  Markakis]{deligkas2022polynomial}
Argyrios Deligkas, Michail Fasoulakis, and Evangelos Markakis.
\newblock A polynomial-time algorithm for 1/3-approximate nash equilibria in
  bimatrix games.
\newblock \emph{arXiv preprint arXiv:2204.11525}, 2022.

\bibitem[Deligkas et~al.(2023)Deligkas, Fasoulakis, and
  Markakis]{deligkas2023polynomial}
Argyrios Deligkas, Michail Fasoulakis, and Evangelos Markakis.
\newblock A polynomial-time algorithm for 1/2-well-supported nash equilibria in
  bimatrix games.
\newblock In \emph{Proceedings of the 2023 Annual ACM-SIAM Symposium on
  Discrete Algorithms (SODA)}, pages 3777--3787. SIAM, 2023.

\bibitem[Fearnley and Savani(2016)]{fearnley2016finding}
John Fearnley and Rahul Savani.
\newblock Finding approximate nash equilibria of bimatrix games via payoff
  queries.
\newblock \emph{ACM Transactions on Economics and Computation (TEAC)},
  4\penalty0 (4):\penalty0 1--19, 2016.

\bibitem[Fearnley et~al.(2013)Fearnley, Gairing, Goldberg, and
  Savani]{fearnley2013learning}
John Fearnley, Martin Gairing, Paul Goldberg, and Rahul Savani.
\newblock Learning equilibria of games via payoff queries.
\newblock In \emph{Proceedings of the fourteenth ACM conference on Electronic
  commerce}, pages 397--414, 2013.

\bibitem[Goldberg and Roth(2016)]{goldberg2016bounds}
Paul~W Goldberg and Aaron Roth.
\newblock Bounds for the query complexity of approximate equilibria.
\newblock \emph{ACM Transactions on Economics and Computation (TEAC)},
  4\penalty0 (4):\penalty0 1--25, 2016.

\bibitem[Karlin and Peres(2017)]{karlin2017game}
Anna~R Karlin and Yuval Peres.
\newblock \emph{Game theory, alive}, volume 101.
\newblock American Mathematical Soc., 2017.

\bibitem[Maiti et~al.(2024)Maiti, Boczar, Jamieson, and Ratliff]{maiti2024near}
Arnab Maiti, Ross Boczar, Kevin Jamieson, and Lillian Ratliff.
\newblock Near-optimal pure exploration in matrix games: A generalization of
  stochastic bandits \& dueling bandits.
\newblock In \emph{International Conference on Artificial Intelligence and
  Statistics}, pages 2602--2610. PMLR, 2024.

\bibitem[v.~Neumann(1928)]{v1928theorie}
J~v.~Neumann.
\newblock Zur theorie der gesellschaftsspiele.
\newblock \emph{Mathematische annalen}, 100\penalty0 (1):\penalty0 295--320,
  1928.

\end{thebibliography}
\clearpage
 


\end{document}